\documentclass[leqno,oneside,english]{article}
\usepackage[T1]{fontenc}
\usepackage[latin1]{inputenc}
\usepackage{amsmath}
\usepackage{amssymb}
\usepackage[english]{babel}

\makeatletter
\usepackage[T1]{fontenc}
\usepackage[latin1]{inputenc}

\usepackage{graphicx}

\makeatletter

\newtheorem{thm}{Theorem}[section]

\newtheorem{defi}[thm]{Definition}

\newtheorem{lem}[thm]{Lemma}

\newtheorem{rk}[thm]{Remark}

\newcommand{\R}{\mathbb{R}}
\newcommand{\Z}{\mathbb{Z}}
\newcommand{\N}{\mathbb{N}}
\newcommand{\tet}{\theta}
\newcommand{\bydef}{\stackrel{\Delta}{=}}
\newcommand{\hinf}{$H_{\infty} \hspace{1mm}$}
\DeclareMathOperator{\MAD}{MAD}
\DeclareMathOperator{\MATI}{MATI}
\DeclareMathOperator{\col}{col}
\newenvironment{proof}[1][Proof]{\textbf{#1.} }{\ \rule{0.5em}{0.5em}}

\makeatother

\usepackage{babel}
\makeatother
\author{Kun Liu \footnote {ACCESS Linnaeus Centre and School of Electrical Engineering,
KTH Royal Institute of Technology, SE-100 44 Stockholm, Sweden, kunliu@kth.se}, Emilia Fridman\footnote {School of Electrical Engineering, Tel Aviv University, Tel Aviv,
69978 Israel, emilia@eng.tau.ac.il}, Laurentiu Hetel\footnote{University Lille Nord de France, LAGIS, FRE CNRS 3303, Ecole Centrale de Lille, Cite Scientifique, BP 48, 59651 Villeneuve d'Ascq cedex, France,
laurentiu.hetel@ec-lille.fr} }

\begin{document}

 \title{Networked control systems in the presence of scheduling protocols and communication delays
 \thanks{This work was partially supported by Israel Science Foundation (grant No 754/10), the Knut and Alice Wallenberg Foundation and the Swedish Research Council.}}

 \maketitle

\begin{abstract}
This paper develops the time-delay approach to Networked Control
Systems (NCSs) in the presence of variable transmission
delays, sampling intervals and communication constraints.  The system sensor nodes are supposed to be distributed over a network.
Due to communication constraints  only one node output is transmitted through the communication channel at once.  The scheduling of sensor information towards the controller is ruled by
a weighted Try-Once-Discard (TOD)  or by Round-Robin (RR) protocols.
Differently from the existing results on NCSs in the presence of scheduling protocols (in the frameworks of hybrid and discrete-time systems), we allow the communication delays  to be greater than the sampling
intervals. A novel  hybrid system model for the closed-loop system is presented
that contains  {\it time-varying delays in the continuous dynamics
and in the reset conditions}.  A new Lyapunov-Krasovskii method,
which is based on discontinuous in time Lyapunov functionals
is introduced for the stability analysis of the delayed
hybrid systems. Polytopic type uncertainties in the system model can be
easily included in the analysis. The efficiency of the  time-delay
approach is illustrated on the examples of uncertain cart-pendulum
and of batch reactor.
\end{abstract}



\section{Introduction}
Networked Control Systems (NCSs) are systems with spatially distributed
sensors, actuators and controller nodes which exchange data over a
communication data channel \cite{Antsaklis_2007}. In many NCSs, only one
node is allowed to use the communication channel at once. The
communication along the data channel is orchestrated by a scheduling
rule called protocol. The introduction of communication network media offers several
practical advantages:  reduced costs, ease of installation and
maintenance and increased flexibility.

 Three main approaches have been used to model the
sampled-data control and  later to the NCSs: a discrete-time \cite{donkers_2011, fujioka_09},  an impulsive/hybrid system \cite{Teel08, Nesic_04} and a time-delay \cite{Fridman92, richard04, GaoSIAM2007, GaoSIAM2010} approaches.
The hybrid system approach, which was inspired by  \cite{Walsh02},
has been applied to nonlinear NCSs under Try-Once-Discard (TOD) and Round-Robin (RR) protocols
in  \cite{HeemelsTAC2010, Nesic_04}.
In the framework of discrete-time approach,
 network-based stabilization of linear
time-invariant systems with TOD/RR protocols and communication delays
has been considered in  \cite{donkers_2011}. 
Variable sampling intervals and/or {\it small communication delays}
(that are {\it smaller than the sampling intervals}) have been
considered in the above works. 

Note that in the absence of scheduling protocols, all the three
approaches are applicable to non-small communication delays (see
e.g., \cite{hetel2010b, Hespanha2010}). The time-delay approach that
was recently suggested in \cite{RR_SCL}
allowed, for the first time, to treat NCSs under  RR protocol
in the presence of non-small communication delays.
In \cite{RR_SCL} the closed-loop system was presented as a switched system  with multiple
ordered delays.

In the present paper, we consider linear (probably, uncertain) NCS with additive
essentially bounded disturbances
in the presence of
scheduling protocols, variable sampling intervals and transmission delays.
Our first goal  is to extend the time-delay
approach to NCSs under TOD protocol in the presence of communication
delays that are allowed to be non-small. This leads to a novel
hybrid system model for the closed-loop system,
where {\it time-varying delays appear in the dynamics and in the
reset equations}. Since a similar hybrid system model
corresponds  to RR protocol, we derive new conditions for Input-To-State (ISS) under
RR protocol as well. These conditions are computationally simpler
than the existing ones of \cite{RR_SCL}
 though may lead to more conservative results.
%
A novel Lyapunov-Krasovskii method is introduced for hybrid delayed
systems, which is based on discontinuous in time Lyapunov
functionals. 

Polytopic type uncertainties in the system
model can be easily included in the analysis.
The efficiency and
advantages of the presented approach are illustrated by two
examples.
Some preliminary results were
presented in \cite{TOD:CDC12}.


{\bf Notation:}\ Throughout the paper, the superscript `$T$' stands
for matrix transposition, $ {\R}^n$ denotes the $n$ dimensional
Euclidean space with vector norm $|\cdot |$, $ {\R}^{n\times m}$ is
the set of all $n\times m$ real matrices, and the notation
$P\!>\!0$, for  $P\in {\R}^{n\times n}$ means that  $P$ is symmetric
and positive  definite. The symmetric elements of the symmetric
matrix will be denoted by ${*}$,
$\lambda_{min}(P)$ denotes the
smallest eigenvalue of matrix $P$. The space of functions
$\phi:[-\tau_M,0] \to  {\R}^n$, which are absolutely continuous on
$[-\tau_M,0]$, 
and have square integrable first order derivatives is denoted by
$W[-\tau_M,0]$ with the norm $\|\phi \|_W=\max_{\tet \in
[-\tau_M,0]}|\phi (\tet )|+ \left [\int_{-\tau_M}^0|\dot
\phi(s)|^2ds \right]^{\frac{1}{2}}.$ ${\Z}_+$ denotes the set \{  0,
1, 2, $\dots$ \}, whereas ${\N}$ denotes the natural numbers. The
symbol $x_t$ denotes $x_t(\theta)=x(t+\theta),$ $\theta\in[-\tau_M,0]$, whereas
$\| w[t_0,t] \|_{\infty}$
 stands for the essential
supremum of the Euclidean norm $| w[t_0,t] |$, where $w:[t_0,t]\to
{\R}^{n_w}$. $\MATI$ and $\MAD$ denote Maximum Allowable Transmission Interval
 and Maximum Allowable Delay, respectively.

\section{Problem formulation}
\subsection{The description of NCS and the hybrid  model}
Consider the system architecture in Figure 1 with plant
\begin{equation}
\label{sys}
\begin{array}{lll}
\dot x(t)=Ax(t)+B u(t)+D \omega(t),\quad t\geq 0,
\end{array}
\end{equation}
where $x(t)\in {\R}^{n}$ is the state vector, $u(t)\in {\R}^{m}$ is
the control input, $\omega(t)\in {\R}^{q}$ is
the essentially bounded disturbance. Assume that there exists a real number  $\Delta>0$
such that  $\|\omega[0, t]\|_{\infty}\leq \Delta$
 for all $t\geq 0$.
  The system matrices $A,$ $B$ and $D$  can be uncertain with polytopic type
uncertainties. 

\begin{figure}[t!]
\centering
\vspace{-3.5cm}
\includegraphics[width=9cm]{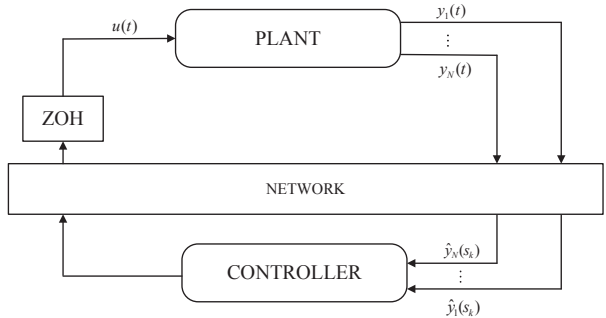}
\vspace{-4.5cm} \caption{\footnotesize System architecture with $N$
sensors}
\end{figure}

The system has several nodes ($N$ distributed sensors, a controller node
and an actuator node) which are connected via  the network.
The measurements are
given by $y_i(t)=C_i x(t) \in  {\R}^{n_i}, i =1, \dots, N, \ \sum_{i=1}^N
n_i=n_y$ and we denote  $C=\left[
\begin{array}{ccccc} C_1^T& \cdots & C_N^T
                          \end{array} \right]^T $, $y(t) =\left[
\begin{array}{ccccc} y_1^T(t)& \cdots & y_N^T(t)
                          \end{array} \right]^T \in  {\R}^{n_y}$.
Let $s_k$ denote the unbounded monotonously increasing sequence
of sampling instants
$$ \label{sequence}
0=s_0< s_1<\cdots< s_k< \cdots, \  \ k \in {\Z_+},\
\lim_{k\rightarrow \infty} s_k=\infty.
$$
%

%
At each sampling instant $s_k$, one of the outputs $y_i(s_k) \in
{\R}^{n_i}$ is transmitted via the sensor network. We
suppose that data loss is not possible and that the transmission of
the information over the network  is subject to a variable
delay $\eta_k$. 
Then $t_k=s_k+\eta_k$ is the updating time instant of the Zero-Order
Hold (ZOH).

Differently from \cite{donkers_2011, HeemelsTAC2010}, we
do not
restrict the network delays to be small 
with $t_k=s_k+\eta_k<s_{k+1}$, i.e. $\eta_{k}<s_{k+1}-s_k$. As in
\cite{Hespanha2010} we allow the delay to be non-small
provided that the old sample cannot get to the destination  (to the
controller  or to the actuator) after the most recent one.
%
%
%
%
Assume that the network-induced delay $\eta_k$ and the time span
between the updating and the most recent sampling instants are
bounded:

\begin{equation}\begin{array}{rr}
\label{tauM} t_{k+1}-t_k+\!\eta_k \leq  \tau_{M} , \ 0\leq
\eta_m\leq \eta_k \leq \MAD, \ k \in {\Z_+},
\end{array}
\end{equation}
%
where $\tau_M$ denotes the maximum time span between the time
\begin{equation}
\label{sktk}s_k=t_k-\eta_k\end{equation}
at which the state is sampled and the time
$t_{k+1}$ at which the next update arrives at the destination. Here
$\eta_m$ and $\MAD$ are known bounds and
$\tau_M= \MATI + \MAD$. 
Since $\MATI=\tau_M-\MAD \leq \tau_M-\eta_m$, $\eta_m > {\tau_M
\over 2}$ implies that the network delays are non-small due to
$\eta_k \geq \eta_m > \tau_M -\eta_m$. In the examples of Section
\ref{sec:examples}, we will show that  our method is applicable for
$\eta_m > {\tau_M \over 2}$.

Denote by $\hat y(s_k)=\left[
\begin{array}{ccccc} \hat{y}_1^T(s_k)& \cdots & \hat{y}_N^T(s_k)
                          \end{array} \right]^T \in  {\R}^{n_y}$ the output information
                          submitted to the scheduling
protocol.
%
At each sampling instant $s_k$, one of the system nodes $i \in
\left\lbrace 1, \dots, N\right\rbrace $ is active, that is only one of
$\hat{y}_i(s_k)$ values is updated with the recent output
${y}_i(s_k)$. Let $i^*_k \in \left\lbrace  1, \dots, N \right\rbrace$ denote the active output node at the sampling instant $s_k$, which will be chosen due to scheduling protocols.

Consider the error between the system output $y(s_k)$ and the last
available information $\hat{y}(s_{k-1})$:
\begin{equation}
\begin{array}{llll}
\label{e:N} e(t)=\col\{e_1(t), \cdots, e_N(t)\}\equiv \hat{y}(s_{k-1})-y(s_k),
 \\ \quad t\in [t_k, t_{k+1}), \ k\in {\Z_+}, \ \hat{y}(s_{-1})\bydef 0, \ e(t)\in {\R}^{n_y}.
 \end{array} \end{equation}
%


We suppose that the controller and the actuator are event-driven (in the sense that the controller and
the ZOH update their outputs as soon as they receive a new sample).
The most recent  output information at the
controller level is denoted by
$\hat y(s_k)$.

\subsubsection*{Static output feedback control}
Assume that there exists a matrix $K=\left[
\begin{array}{ccc} K_1 & \cdots & K_N\end{array} \right]$, $ K_i \in  {\R}^{m \times n_i}$ such that  $A+BKC$ is Hurwitz.
Then, the static output feedback controller has a form
\begin{equation}
\label{control:quan:TOD:N}\begin{array}{lll}
 u(t)= K_{i_k^*} y_{i_k^*}{(t_k-\eta_k)}\!+\!\sum_{i=1, i \neq i_k^*}^N K_{i} \hat
y_i{(t_{k-1}\!-\!\eta_{k-1}}), \ t\in[t_k,t_{k+1}), 
\end{array}\end{equation}
where $i_k^*$ is the index of the active node at $s_k$ and $\eta_k$ is communication delay.
We obtain thus the impulsive closed-loop model with the following
continuous dynamics:
\begin{equation}
\begin{array}{lll} \label{x:N}
\dot x(t)=Ax(t)+ A_1x(t_k-\eta_k)+ \sum_{i=1, i \neq i_k^*}^N  B_{i} e_{i}(t)+D \omega(t),  \\
\dot e(t)=0, \quad t\in [t_k,t_{k+1}),
\end{array}
\end{equation}
where $A_1=BKC, \ B_{i}=BK_{i}, \ i=1, \dots, N.$ 
%

Taking into account (\ref{e:N}), we obtain
$$
\begin{array}{lll}
e_i(t_{k+1})&=&\hat y_i(s_k)-y_i(s_{k+1})=y_i(s_k)-y_i(s_{k+1})\\
&=&C_ix(s_k)-C_ix(s_{k+1}), \ i={i_k^*},
\end{array}
$$
and
$$
\begin{array}{lll}
e_{i}(t_{k+1})&=&\hat y_{i}(s_k)-y_{i}(s_{k+1})=\hat y_{i}(s_{k-1})-y_{i}(s_{k+1})\\
&=&\hat y_{i}(s_{k-1})-y_{i}(s_k)+y_{i}(s_k)-y_{i}(s_{k+1})\\
&=&e_{i}(t_{k})+C_{i}[x(s_k)-x(s_{k+1})], \ i \neq i_k^*,
i\in \N. 
\end{array}
$$
Thus, the delayed  reset system is given by
\begin{equation}
\begin{array}{lll} \label{reset:N}
x(t_{k+1})=x(t_{k+1}^-),\\
e_i(t_{k+1})=C_i[x(t_k-\eta_k)\!-\!x(t_{k+1}-\eta_{k+1})], \ i={i_k^*},\\
e_{i}(t_{k+1})\!=\!e_{i}(t_{k})\!+\!C_{i}[x(t_k-\eta_k)\!-\! x(t_{k+1}-\eta_{k+1})],\ i \!\neq\! i_k^*,
i\in \N.
\end{array}
\end{equation}

Therefore, (\ref{x:N})-(\ref{reset:N}) is the hybrid model of the NCS.
Since  $x(t_k-\eta_k)=x(t-\tau(t))$  for $t\in [t_k, t_{k+1})$ with
$\tau(t)=t-t_k+\eta_k \in [\eta_m, \tau_M]$ (cf. (\ref{tauM})), the
hybrid system model (\ref{x:N})-(\ref{reset:N}) contains the
piecewise-continuous delay $\tau(t)$ in the continuous-time dynamics
(\ref{x:N}). Even for $\eta_k=0$ we  have the delayed
state $x(t_k)=x(t-\tau(t))$ with $\tau(t)=t-t_k$.

Note that the first updating time $t_0$ corresponds to the
time instant when the first data is received by the actuator.
 Assume that initial
conditions for (\ref{x:N})-(\ref{reset:N}) are given by
$x_{t_0}\in W[-\tau_M,0]$ and $e(t_0)=-Cx(t_0-\eta_0)=-Cx_0$.


\subsubsection*{Dynamic output feedback control}
Assuming that the controller is directly connected to the
actuator, 
consider a dynamic output feedback controller of the form
$$ 
\label{controller:D:N}
\begin{array}{lll}
\dot x_c(t)=A_cx_c(t)+B_c \hat y(s_k),\\
u(t)=C_cx_c(t)+D_c \hat y(s_k), \ \ t\in[t_k, t_{k+1}), \ k\in {\Z_+},
\end{array}
$$
where $x_c(t)\in  {\R}^{n_c}$, 
$A_c, B_c, C_c$ and $D_c$ are the matrices with appropriate
dimensions. Let $e_i(t) (i=1, \dots, N)$ be defined by (\ref{e:N}).
The closed-loop system 
can be presented in the form of (\ref{x:N})-(\ref{reset:N}), where
$x$, $e_i$ and the
matrices are changed by the ones with the bars as follows:
$$
\begin{array}{lll}
\bar x=\left[\begin{array}{lc} x \\ x_c\end{array}\right],\
\bar A=\left[\begin{array}{lc}A &  BC_c\\
0_{n_c \times n} &  A_c\end{array}\right],  \bar B_i=\left[\begin{array}{c}BD_c  \\
B_c \end{array}\right],  \bar D=\left[\begin{array}{c}D  \\
0_{n_c \times q} \end{array}\right],\\
\bar A_1=\scriptsize \left[\begin{array}{cl}BD_cC  & 0_{n \times n_c}\\
B_c C & 0_{n_c \times n_c}\end{array}\right],
\bar C_1=\left[
\begin{array}{cc} C_1& 0\\
0 & 0
                          \end{array} \right], \bar C_i\in  {\R}^{\!n_y\! \times \!(n+n_c)},  \\
\bar C_2=\left[
\begin{array}{lcc} 0_{n \times n_1} & C_2^T& 0\\
0_{n_c \times n_1} & 0 & 0
                          \end{array} \right]^T, \dots ,
\bar C_N=\left[
\begin{array}{cc}  0 & C_N^T\\
 0 & 0
                          \end{array} \right]^T, \\
 %
%
\bar e_1(t)\!=\![e_1^T(t) \ 0]^T, \bar e_2(t)\!=\![0_{1 \times n_1 } \ e_2^T(t) \ \ 0]^T,  \cdots \
\bar e_N(t)\!=\![0 \ \ e_N^T(t)]^T,  \bar e_i(t)\in  {\R}^{n_y}.
\end{array}
$$

\begin{rk}\label{model:compare}
In \cite{HeemelsTAC2010}, a  
piecewise-continuous error $e(t)=\hat{y}(t_k)-y(t), t\in [s_k,
s_{k+1}]$ is defined, which leads to the non-delayed continuous
dynamics.  The derivation of reset equations is based on
the assumption of small communication delays, that is avoided in our
approach.
In our  approach $e(t)$ is different: it is given by (\ref{e:N}) and is
piecewise-constant. As a result, our hybrid model is different with the  delayed continuous dynamics.
Moreover, in the absence of scheduling protocols, the closed-loop system is  given by non-hybrid
system
 (\ref{x:N}), where $e(t)\equiv 0$. The latter is consistent
 with the time-delay model considered e.g. in \cite{GaoSIAM2007, GaoSIAM2010}.
\end{rk}

\subsection{Scheduling protocols}
\subsubsection*{TOD protocol}\label{sec_TODpr}
In TOD protocol, the output node $i \in \left\lbrace 1, \dots, N\right\rbrace $ with
the greatest (weighted) error will be granted the access to the
network.

\begin{defi}[Weighted TOD protocol]
Let $Q_i>0 (i=1, \dots, N)$  be some weighting matrices. At the sampling instant $s_k$, the \textit{weighted
TOD protocol} is a protocol for which the active output node with the index $i_k^*$
 is defined as any index  that satisfies
\begin{equation}
\begin{array}{lll} \label{if:N}
|\sqrt{Q}_{i_k^*}e_{i_k^*}(t)|^2\geq |\sqrt{Q}_{i}e_{i}(t)|^2, \
t\in [t_k, t_{k+1}), \ k\in {\Z_+}, \ i=1, \dots, N.
\end{array}
\end{equation}
\end{defi}
A possible choice of $i_k^*$ is given by
$$
 i^*_k = \min \{\arg \max_{i \in \left\lbrace  1, \dots, N \right\rbrace}
  |\sqrt{Q}_i \left(\hat{y}_i(s_{k-1})-y_i(s_k) \right) |^2\}.
$$
The conditions for computing the weighting matrices $Q_1,\dots,Q_N$
will be given in Theorem \ref{thm:ISS:TOD:N} below.

\begin{rk}
 For implementation of
 TOD protocol in wireless networks we refer to \cite{ChrisTODrealization}.
\end{rk}

\subsubsection*{RR protocol}\label{sec:RR:protocol}
The  active output node is chosen periodically:
\begin{equation}
\begin{array}{lll} \label{RR:ik:N}
i_k^*=i_{k+N}^*, \ for \ all \ k \in {\Z_+},\\
i_j^* \neq i_l^*, \ for \ 0\leq j<l\leq N-1.
\end{array}
\end{equation}
\begin{rk}\label{model:RR:N}
Note that another model for the closed-loop system under RR protocol
was given in \cite{RR_SCL}.
The model in \cite{RR_SCL} is a switched system with ordered delays
$\tau_1(t)<\cdots<\tau_N(t)$, where $\tau_i(t)=t-t_{k-i+1}+\eta_{k-i+1}, i=1, \dots, N$.
A Lyapunov-Krasovskii analysis of the latter model is based on the standard time-independent Lyapunov functional
for interval  delay.
\end{rk}

\begin{defi}
The hybrid system (\ref{x:N})-(\ref{reset:N}) with essentially
bounded disturbance $\omega$ is said to be
partially ISS with respect to $x$ (or $x$-ISS) if there exist
constants $b
>0, \delta>0$ and $c>0$ such that the following holds for $t\geq
t_0$
$$\begin{array}{ll}
\label{ISS} |x(t)|^2\leq be^{-\delta (t-t_0)}\left[ 
\|x_{t_0}\|_W^2+|e(t_0)|^2\right] +c\|\omega[t_0, t]\|_{\infty}^2
\end{array}
$$
for the solutions of the hybrid system initialized with
$x_{t_0}=\phi\in { W[-\tau_M,0]}$
 and $e(t_0)\in {\R}^{n_y}$. The hybrid system
(\ref{x:N})-(\ref{reset:N}) is ISS if additionally the following bound is valid for $t\geq t_0$
$$|e(t)|^2\leq be^{-\delta (t-t_0)}\left[ 
\|x_{t_0}\|_W^2+|e(t_0)|^2\right] +c\|\omega[t_0, t]\|_{\infty}^2.$$
\end{defi}

Our objective  is to derive Linear Matrix Inequality (LMI) conditions for the partial ISS of
the hybrid system (\ref{x:N})-(\ref{reset:N}) with respect to the
variable of interest $x$. In \cite{NesicAut07Partial}, the notion of
partial stability was also used. In Section \ref{ISS:TOD:N} below,
 ISS of (\ref{x:N})-(\ref{reset:N}) under TOD protocol with $N$ sensor nodes
 will be studied. For $N=2$,
less restrictive conditions will be derived in Section
\ref{ISS:TOD:RR:N:2}, and it will be shown that the same conditions
guarantee $x$-ISS of (\ref{x:N})-(\ref{reset:N}) under RR protocol.
In Section \ref{ISS:RR:N}, the latter conditions  will be extended
to RR protocol with  $N\geq 2$.

\section{ISS under TOD protocol: general $N$}\label{ISS:TOD:N} 

Note that the  differential equation for $x$ given by
(\ref{x:N}) depends  on $e_i(t)=e_i(t_k),\ t\in[t_k,t_{k+1})$ with $i\neq i_k^*$ only.
 Consider the following Lyapunov functional:
\begin{equation}\label{V:TOD:N}
\begin{array}{lll}
V_e(t)=V(t, x_t, \dot x_t)+ \sum_{i=1}^Ne_i^T(t)Q_ie_i(t),\\
V(t,x_t, \dot x_t)= \tilde V(t, x_t, \dot x_t)+ V_G,\\
V_G=\sum_{i=1}^N(\tau_M-\!\!\eta_m)\int_{s_k}^{t}e^{2\alpha(s-t)}|\sqrt{G_i}C_i \dot x(s)|^2ds,\\
\tilde V(t,x_t, \dot x_t)=x^T(t)P x(t)+\int_{t-\eta_m}^{t}
e^{2\alpha(s-t)} x^T(s) S_0 x(s) ds \\
\quad +\int_{t-\tau_M}^{t-\eta_m}
e^{2\alpha(s-t)}x^T(s) S_1 x(s) ds\\
\quad + \eta_m \int_{- \eta_m}^{0}\int_{t+\tet }^t
 e^{2\alpha(s-t)}\dot x^T(s) R_0\dot x(s) ds d\tet\\
\quad +  (\tau_M-\eta_m)\int_{- \tau_M}^{-\eta_m}\int_{t+\!\tet }^t
 e^{2\alpha(s-t)}\dot x^T(s) R_1\dot x(s) ds d\tet,\\
 P>0, S_j>0,  R_j>0, G_i>0, Q_i>0, \alpha >0,\\ j=0,1, i=1,\dots,N,
\ t\in [t_k, t_{k+1}), \ k\in {\Z_+}, \end{array}\end{equation}
where $x_t(\theta)\bydef x(t+\theta), \ \theta\in[-\tau_M,0]$
and where we define (for simplicity) $x(t)=x_0, \ \ t<0.$

Here  the terms
$$e_i^T(t)Q_ie_i(t)\equiv
e_i^T(t_k)Q_ie_i(t_k), \ t\in [t_k,t_{k+1})$$ are
piecewise-constant, $\tilde V(t,x_t, \dot x_t)$ presents the
standard Lyapunov functional for systems
 with  interval delays $\tau(t)\in [\eta_m, \tau_M]$.
{\it The novel piecewise-continuous in time term} $V_G$
 {\it is inserted to cope with the delays in the reset
conditions}. It is continuous on $[t_k,t_{k+1})$ and
 do not grow in the jumps (when $t=t_{k+1}$), since
\begin{equation}\label{Q:N}
\begin{array}{lll}
{V_G}_{|t=t_{k+1}}-{V_G}_{|t=t_{k+1}^-}\\
=\sum_{i=1}^N(\tau_M-\eta_m) \int_{s_{k+1}}^{t_{k+1}} e^{2 \alpha (s-t_{k+1})} |\sqrt {G_i}C_i\dot x (s)|^2ds\\
\ \ \ -\sum_{i=1}^N(\tau_M-\eta_m)\int_{s_{k}}^{t_{k+1}^-}
e^{2 \alpha(s-t_{k+1})} |{\sqrt {G_i}C_i}\dot x(s)|^2 ds\\
\leq -\sum_{i=1}^N(\tau_M-\eta_m)e^{-2\alpha \tau_M}
\int_{s_{k}}^{s_{k+1}}|\sqrt {G_i}C_i\dot
x(s)|^2ds\\
\leq - \sum_{i=1}^Ne^{-2\alpha \tau_M} |\sqrt
{G_i}C_i[x({s_{k}})-x({s_{k+1}})]|^2,\end{array}\end{equation}
where we  applied Jensen's inequality (see e.g., \cite{Gu03}). The function $V_e(t)$ is thus continuous and
differentiable over $[t_k,t_{k+1})$.
 The following lemma gives sufficient conditions for
the $x$-ISS of (\ref{x:N})-(\ref{if:N}):

\begin{lem} \label{lemma:TOD:N} 
 Let there exist  positive constants $\alpha,$ $b,$ $0<Q_i\in {\R}^{n_i\times n_i}$, $0<U_i\in {\R}^{n_i\times n_i}$, $0<G_i\in {\R}^{n_i\times n_i},$ $i=1,\dots, N,$ and $V_e(t)$ of  (\ref{V:TOD:N}) such that
 along (\ref{x:N}) the following inequality holds
\begin{equation}
\begin{array}{ll}\label{inequality1:lemma:TOD:N}
\dot {V_e}(t)+2\alpha V_e(t)-{1\over \tau_M-\eta_m} \sum_{i=1, i \neq
i_k^*}^N|\sqrt{U_i}e_i(t)|^2\\
\ \ -2\alpha|\sqrt{Q_{i_k^*}}e_{i_k^*}(t)|^2- b |\omega(t)|^2\leq 0,\ t\in[t_k, t_{k+1}).\\
\end{array}\end{equation}
Assume additionally that
\begin{equation}
\begin{array}{rr}\label{Q:Q12:N}\Omega_i \!\bydef \!
\begin{bmatrix}
-{1-2\alpha (\tau_M-\eta_m)\over N-1}Q_i+U_i & Q_i \\
 {*} & Q_i  - G_i{e^{-2\alpha \tau_M}}\end{bmatrix}< 0 ,\  i=1,\dots, N.
\end{array}\end{equation}
Then $V_e(t)$ does not grow in the jumps along
(\ref{x:N})-(\ref{if:N}):
\begin{equation}
\begin{array}{llll}\label{jump:tk1:N}
\Theta \bydef V_e(t_{k+1})-V_e(t_{k+1}^-) +\sum_{i=1, i \neq
i_k^*}^N|\sqrt{U_i}e_i(t_k)|^2\\
\hspace{.7cm}+2\alpha(\tau_M-\eta_m)|\sqrt{Q_{i_k^*}}e_{i_k^*}(t_k)|^2\leq 0.
\end{array}\end{equation}
Moreover, the following bounds hold for a solution of
(\ref{x:N})-(\ref{if:N}) initialized by  $x_{t_0}\in W[-\tau_M,0],
e(t_0)\in {\R}^{n_y}$:
\begin{equation}\label{lem:V:bound:N}
\begin{array}{llll}V(t, x_t, \dot x_t)\leq e^{-2\alpha (t-t_0)} V_e(t_0)
+{b \over  2\alpha} \Delta^2, \ t\geq t_0, \\
 V_e(t_0)=V(t_0,
x_{t_0}, \dot
x_{t_0})+\sum_{i=1}^N|\sqrt{Q_i}e_i(t_0)|^2,
\end{array}\end{equation}
and
\begin{equation}\label{lem:V:bound:eN}
\begin{array}{llll}
\sum_{i=1}^N|\sqrt{Q_i}e_i(t)|^2\leq \tilde c e^{-2\alpha (t-t_0)}
V_e(t_0) +{b \over  2\alpha}  \Delta^2,\\
\end{array}\end{equation}
 where $\tilde c=e^{2\alpha(\tau_M-\eta_m)}$, implying ISS of
(\ref{x:N})-(\ref{if:N}).
\end{lem}
\begin{proof}
Since $\int_{t_k}^{t}e^{-2\alpha(t-s)}ds \leq \tau_M-\eta_m,$
$t\in[t_k, t_{k+1})$ and
 $|\omega(t)| \leq  \Delta$, by the comparison principle,
(\ref{inequality1:lemma:TOD:N}) implies
\begin{equation}\begin{array}{llll}\label{lem1:N}
V_e(t)\leq e^{-2\alpha(t-t_k)}V_e(t_k)+\sum_{i=1, i \neq i_k^*}^N\{|\sqrt{U_i}e_i(t_k)|^2\}\\
\hspace{1.1cm}+2\alpha (\tau_M-\eta_m) |\sqrt{Q_{i_k^*}}e_{i_k^*}(t_k)|^2+ b\Delta^2\int_{t_k}^te^{-2\alpha (t-s)}ds,
\ t\in[t_k, t_{k+1}).
\end{array}\end{equation}
Note that (\ref{Q:Q12:N}) yields $0<2\alpha (\tau_M-\eta_m)<1$ and $U_i \leq {1-2\alpha (\tau_M-\eta_m)\over N-1}Q_i \leq Q_i, \ i=1,\dots, N$. Hence,
\begin{equation}\begin{array}{llll}\label{lem1a:N}
V(t, x_t, {\dot x}_t)\leq e^{-2\alpha (t-t_k)}V_e(t_k)+b \Delta^2\int_{t_k}^te^{-2\alpha (t-s)}ds,
\ t\in[t_k, t_{k+1}).
\end{array}\end{equation}
Since $\tilde V_{|t=t_{k+1}}=\tilde V_{|t=t_{k+1}^-}$ and
$e(t_{k+1}^-)=e(t_k)$,
 we obtain
$$
\begin{array}{llllll}
\Theta=\sum_{i=1}^N[|\sqrt{Q_i}e_i(t_{k+1})|^2-|\sqrt{Q_i}e_i(t_{k})|^2]
+\sum_{i=1, i \neq i_k^*}^N|\sqrt{U_i}e_i(t_k)|^2 \\ \hspace{.7cm}+2\alpha(\tau_M-\eta_m)\sqrt{Q_{i_k^*}}e_{i_k^*}(t_k)|^2+{V_G}_{|t=t_{k+1}}-{V_G}_{|t=t_{k+1}^-}.
\end{array}$$
Then taking into account (\ref{Q:N}) we find
$$\begin{array}{llllll}
\Theta\leq |\!{\sqrt Q_{i^*_k}}e_{i^*_k}(t_{k+1})|^2\!\!+\!\!\sum_{i=1, i\neq i_k^*}^N|\!\sqrt{Q_i}e_i(t_{k+1})|^2\!\!-\![1\!-\!2\alpha (\tau_M\!\!-\!\!\eta_m)]|\!\sqrt{Q_{i^*_k}}e_{i^*_k}(t_{k})|^2\\
\hspace{.5cm}-\sum_{i=1, i\neq i^*_k}^N\{e_i^T(t_{k})[Q_i\!-\!U_i]e_i(t_{k})\}\!-\!\sum_{i=1}^Ne^{-2\alpha \tau_M} |\sqrt
{G_i}C_i[x(s_{k})\!-\!x(s_{k+1}
)]|^2.\\
\end{array}$$
Note that under TOD protocol
$$\begin{array}{llllll}
-|\sqrt{Q_{i^*_k}}e_{i^*_k}(t_{k})|^2\leq -{1\over N-1}\sum_{i=1,i\neq i_k^*}^N|\sqrt{Q_i}e_i(t_{k})|^2.
\end{array}$$
Denote $\zeta_i=\col\{e_i(t_k), C_i[x({s_{k}})-x({s_{k+1}
})]\}.$ Then, employing (\ref{reset:N}) and \eqref{sktk} we arrive at
$$\begin{array}{llllll}
\Theta \leq -|\sqrt {G_{i_k^*}e^{-2\alpha \tau_M}\!-\!Q_{i_k^*}}C_{i_k^*}[x({s_{k}})-x({s_{k+1} })]|^2 +\sum_{i=1, i \neq i_k^*}^N\zeta_i^T\Omega_i \zeta_i \leq 0,
\end{array}$$
that yields (\ref{jump:tk1:N}).

The inequalities (\ref{jump:tk1:N}) and (\ref{lem1:N}) with $t=t_{k+1}^-$ imply
$$\begin{array}{llll}
V_e(t_{k+1}) \leq  e^{-2\alpha (t_{k+1}-t_k)}V_e(t_k)+b\Delta^2\int_{t_k}^{t_{k+1}}e^{-2\alpha (t_{k+1}-s)}ds.
\end{array}$$ 
Then

\begin{equation}
\label{Vk2:N}
\begin{array}{llll}
V_e(t_{k+1}) \leq
 e^{-2\alpha
(t_{k+1}-t_{k-1})}V_e(t_{k-1}) + b\Delta^2 \int_{t_{k-1}}^{t_{k+1}}e^{-2\alpha (t_{k+1}-s)}ds\\
\hspace{1.4cm} \leq  e^{-2\alpha
(t_{k+1}-t_{0})}V_e(t_{0})+ b\Delta^2 \int_{t_0}^{t_{k+1}}e^{-2\alpha (t_{k+1}-s)}ds.
\end{array}
\end{equation}
Replacing in \eqref{Vk2:N} $k+1$  by $k$ and using
(\ref{lem1a:N}), we arrive at (\ref{lem:V:bound:N}), which yields
$x$-ISS of (\ref{x:N})-(\ref{if:N}) because
$$\lambda_{min}(P)|x(t)|^2 \leq V(t, x_t, \dot x_t), \  V(t_0,
x_{t_0}, \dot x_{t_0}) \leq \delta \|x_{t_0}\|_W^2$$ for some scalar
$\delta>0$.
Moreover, \eqref{Vk2:N} with $k+1$ replaced by $k$ implies
(\ref{lem:V:bound:eN}) since for $t\in[t_k, t_{k+1})$
$$e^{-2\alpha (t_k-t_0)}=e^{-2\alpha
(t-t_0)}e^{-2\alpha (t_k-t)}\leq \tilde c e^{-2\alpha (t-t_0)}.$$
\end{proof}


By using Lemma \ref{lemma:TOD:N}
and the standard arguments for the delay-dependent analysis, we
derive LMI conditions for ISS of (\ref{x:N})-(\ref{if:N}) (see
Appendix A for the proof):
\begin{thm}

\label{thm:ISS:TOD:N}   Given $ 0 \leq \eta_m < \tau_M,$ $\alpha >0$,
assume that there exist positive scalar $b$, $n\times n$ matrices
$P>0$, $S_0>0$, $R_0>0$, $S_1>0$, $R_1>0$, $S_{12}$, $n_i\times n_i$ matrices $Q_i>0$,
$U_i>0,$ $G_i>0,$ $i=1,\dots, N,$ such that (\ref{Q:Q12:N}) and the following inequalities are
feasible:
\begin{equation}
\label{LMI1:thm:ISS:TOD:N}
\Phi=\left[\begin{array} {ccc}
R_1 & S_{12}\\
{*} & R_1
\end{array} \right] \geq 0,
\end{equation}
\begin{equation}
\begin{array}{llll}
\label{LMI2:thm:ISS:TOD:N}
 \begin{bmatrix}
 \Sigma_i- (F^i)^T\Phi   F^ie^{-2\alpha \tau_M}&\ \Xi_i^T H \\
 {*} &\ -H\end{bmatrix} <0, \ i=1,\dots, N,
\end{array}
\end{equation}
where
\begin{equation}\begin{array}{llll}
\label{thm1:phi:ISS:TOD:N}
H=\eta_m^2 R_0+(\tau_M-\eta_m)^2R_1+(\tau_M-\eta_m)\sum_{l=1}^N C_l^T G_lC_l,\\
 \Sigma_i= (F_1^i)^TP \Xi_i+ (\Xi^i)^TP F_1^i+ \Upsilon_i- (F_2^i)^TR_0 F_2^ie^{-2\alpha \eta_m},\\
 F_1^i=[I_{n} \ 0_{n \times (3n+n_y-n_i+q)}], \\
 F_2^i=[I_{n} \ -I_{n} \ \ 0_{n \times (2n+n_y-n_i+q)}],\\
 F^i=\left[\begin{array} {ccccc}
0_{n \times n} & I_{n} & -I_{n} & 0_{n \times n} &0_{n \times (n_y-n_i+q)}\\
0_{n \times n} & 0_{n \times n} & I_{n} & -I_{n } &0_{n \times (n_y-n_i+q)}
\end{array} \right],\\
\Xi_i=[A \ 0_{n \times n} \ A_1 \ 0_{n \times n} \ B_2 \ \cdots B_N \ D],\ i=1,\\
\Xi_i=[A \ 0_{n \times n} \ A_1 \ 0_{n \times n} \ B_1 \ \cdots \ {B_j}_{|j\neq i} \cdots B_N \ D],\ i=2,\dots N,\\
\Upsilon_i\!=\!{\rm diag}\{S_0\!\!+\!\!2 \alpha P, -(S_0\!-\!S_1)e^{-2\alpha \eta_m},  0_{n \times n},
 -S_1e^{-2\alpha \tau_M}, \psi_{2}, \!\cdots, \psi_N, \!-bI_q \},\ i=1,\\
 \Upsilon_i\!=\!\!{\rm diag}\{S_0\!\!+\!\!2 \alpha P, -(S_0\!\!-\!S_1)e^{-2\alpha \eta_m},  0_{n \times n},
 -\!S_1e^{-2\alpha \tau_M},\psi_{1}, \!\!\cdots,\!\! {\psi_j}_{|j\neq i}, \!\!\cdots,\!\! \psi_N, \!-bI_q \},\\ i=2,\dots N,\ \psi_j=- {1 \over \tau_M-\eta_m}U_j+2 \alpha Q_j, \ j=1, \dots, N.
\end{array}\end{equation}
%
Then solutions of the hybrid system (\ref{x:N})-(\ref{if:N}) satisfy
the bound (\ref{lem:V:bound:N}),
where $V(t, x_t, \dot x_t)$ is given by (\ref{V:TOD:N}), implying
ISS of (\ref{x:N})-(\ref{if:N}).  If the above LMIs are feasible
with $\alpha=0$, then the bound (\ref{lem:V:bound:N}) holds with a
small enough $\alpha_0>0$.
\end{thm}


\section{ISS under TOD/RR protocol: $N=2$}\label{ISS:TOD:RR:N:2}
For $N=2$ less restrictive conditions than those of Theorem
\ref{thm:ISS:TOD:N}
 for the $x$-ISS of (\ref{x:N})-(\ref{reset:N}) will be derived via a different
 from (\ref{V:TOD:N}) Lyapunov functional:
\begin{equation}\label{V:TOD}
\begin{array}{lll}
V_e(t)=V(t, x_t, \dot x_t)+{t_{k+1}-t\over \tau_M-\eta_m}\{e_i^T(t)Q_ie_i(t)\}_{|i\neq i^*_k},\\
Q_1>0, \ Q_2>0, \ \alpha >0, \ t\in [t_k, t_{k+1}), \ k\in {\Z_+},
\end{array}\end{equation} where $i^*_k \in \{1,2 \}$ and $ V(t,
x_t, \dot x_t)$ is given by (\ref{V:TOD:N}) with
$G_i=Q_ie^{2\alpha \tau_M}.$
The term ${t_{k+1}-t\over \tau_M-\eta_m}\{e_i^T(t_k)Q_ie_i(t_k)\}$
is inspired by the similar construction of Lyapunov functionals for
the sampled-data systems \cite{samp, Teel08, Seuret_12}. The following statement holds:

\begin{lem} \label{lemma:TOD} Given $N=2$, 
if there exist  positive constants $\alpha,$ $b$ and $V_e(t)$ of  (\ref{V:TOD}) such that
 along (\ref{x:N})-(\ref{if:N}) ((\ref{x:N}), (\ref{reset:N}), (\ref{RR:ik:N})) the following inequality holds
\begin{equation}
\begin{array}{llll}\label{inequality1:lemma:TOD}
\dot {V_e}(t)+2\alpha V_e(t)
 -b |\omega(t)|^2\leq  0
,\quad t\in[t_k, t_{k+1}).
\end{array}\end{equation}
%
Then $V_e(t)$ does not grow in the jumps along (\ref{x:N})-(\ref{if:N}) ((\ref{x:N}), (\ref{reset:N}), (\ref{RR:ik:N})), where
\begin{equation}
\begin{array}{llll}\label{jump:tk1}
\Theta \bydef V_e(t_{k+1})-V_e(t_{k+1}^-) 
\leq 0.
\end{array}\end{equation}
The bound (\ref{lem:V:bound:N}) is valid for a
solution of (\ref{x:N})-(\ref{if:N}) ((\ref{x:N}), (\ref{reset:N}), (\ref{RR:ik:N})) with the initial
condition $x_{t_0}\in W[-\tau_M,0]$, $e(t_0)\in {\R}^{n_y}$,
implying the $x$-ISS of  (\ref{x:N})-(\ref{if:N}) ((\ref{x:N}), (\ref{reset:N}), (\ref{RR:ik:N})).

\end{lem}

\begin{proof}
Since $|\omega(t)| \leq  \Delta$, (\ref{inequality1:lemma:TOD}) implies
\begin{equation}\begin{array}{llll}\label{lem1}
V_e(t)\leq e^{-2\alpha(t-t_k)}V_e(t_k)+b \Delta^2\int_{t_k}^te^{-2\alpha (t-s)}ds,
\ t\in[t_k, t_{k+1}).
\end{array}\end{equation}
 Noting that
$$\begin{array}{ll}
V_e(t_{k+1})\leq \tilde V_{|t=t_{k+1}}+|\sqrt{Q}_ie_i(t_{k+1})|_{|i\neq i^*_{k+1}}^2\\
\hspace{1.7cm}+\sum_{i=1}^2(\tau_M-\eta_m)
\int_{t_{k+1}-\eta_{k+1}}^{t_{k+1}}e^{2\alpha(s-t_{k+1})} |{\sqrt {G_i}C_i}\dot x(s)|^2ds,
\end{array}$$
 we obtain employing \eqref{Q:N}
$$
\begin{array}{llllll}
\Theta \leq e_i^T(t_{k+1})Q_ie_i(t_{k+1})_{|i\neq i^*_{k+1}}\!+\!{V_G}_{|t=t_{k+1}}\!-\!{V_G}_{|t=t_{k+1}^-}\\
\quad \leq e_i^T(t_{k+1})Q_ie_i(t_{k+1})_{|i\neq i^*_{k+1}}-\sum_{i=1}^2 |{\sqrt {Q_i}C_i}[x({t_{k}-\eta_{k}})-x({t_{k+1} -\eta_{k+1}})]|^2.
\end{array}$$
We will prove that $\Theta\leq 0$ under TOD and RR protocols,
respectively. Under TOD protocol we have
$$
\begin{array}{llllll}
e_i^T(t_{k+1})Q_ie_i(t_{k+1})_{|i\neq i^*_{k+1}} \leq e_{i_k^*}^T(t_{k+1})Q_{i_k^*}e_{i_k^*}(t_{k+1})
\end{array}$$
for $i^*_{k+1}=i^*_{k}$, whereas
\begin{equation}\label{lem:RR}
\begin{array}{llllll}
e_i^T(t_{k+1})Q_ie_i(t_{k+1})_{|i\neq i^*_{k+1}}=
e_{i_k^*}^T(t_{k+1})Q_{i_k^*}e_{i_k^*}(t_{k+1})
\end{array}\end{equation}
for $i^*_{k+1}\neq i^*_{k}$. Then, taking into account
(\ref{reset:N}) we obtain
$$
\begin{array}{llllll}
\Theta \leq |{\sqrt {Q_i}C_i}[x({t_{k}-\eta_{k}})-x({t_{k+1} -\eta_{k+1}})]|^2_{|i=i_k^*}\\
\quad  -\sum_{i=1}^2 |{\sqrt
{Q_i}C_i}[x({t_{k}-\eta_{k}})-x({t_{k+1} -\eta_{k+1}})]|^2 \leq
0.\end{array}$$
Under RR protocol we have $i^*_{k+1}\neq i^*_{k}$ meaning that
(\ref{lem:RR}) holds and that  $\Theta\leq 0$. Then the result follows by the arguments of
 Lemma \ref{lemma:TOD:N}. \end{proof}

\begin{rk}
Differently from Lemma \ref{lemma:TOD:N},  Lemma \ref{lemma:TOD}
guarantees \eqref{Vk2:N} that does not give a bound on $e_{i_k^*}(t_k)$
since $V_e(t)$ for $t\in[t_k, t_{k+1})$ does not depend on
$e_{i_k^*}(t_k)$.
  That is why  Lemma \ref{lemma:TOD} guarantees only $x$-ISS.
  However, as explained in Remark \ref{rem_ebound} below, under RR protocol $x$-ISS implies boundedness of $e$.
\end{rk}

In the next section, we will extend the result of Lemma \ref{lemma:TOD} under RR protocol to the case of $N\geq 2$.
Theorem  \ref{thm:ISS:RR:N} below (in the particular case of $N=2$) will provide
LMIs for the
$x$-ISS of (\ref{x:N})-(\ref{if:N}) ((\ref{x:N}), (\ref{reset:N}), (\ref{RR:ik:N})).

%
%

\section{ISS under RR protocol: $N\geq 2$}\label{ISS:RR:N}


Under RR protocol (\ref{RR:ik:N}), the reset system (\ref{reset:N}) can be rewritten as
\begin{equation}
\begin{array}{llll} \label{RRreset:N}
x(t_{k+1})=x(t_{k+1}^-),\\
e_{i_{k-j}^*}(t_{k+1})=C_{i_{k-j}^*}[x(s_{k-j})-x(s_{k+1})],\\
j=0,\dots,N-1 \ {\rm if} \  k\geq  N-1,
\end{array}
\end{equation}
where the index $k-j$ corresponds to the last updated measurement in the node $i_{k-j}^*$.

Consider the following Lyapunov functional:
\begin{equation}\label{V:RR}
\begin{array}{lll}
V_e(t)=V(t, x_t, \dot x_t)+V_Q,\quad t\geq t_{N-1},\\
V(t, x_t, \dot x_t)=\tilde V(t, x_t, \dot x_t)+V_G,\\
\end{array}\end{equation}
where $\tilde V(t, x_t, \dot x_t)$ is given by (\ref{V:TOD:N}). The
discontinuous in time terms $V_Q$ and $V_G$ are defined as follows:
\begin{equation}\begin{array}{lllll}
\label{initial:VG}
V_Q=
\sum_{ j=1}^{N-1}{t_{k+1}-t\over
j(\tau_M-\eta_m)}|{\sqrt{Q_{i_{k-j}^*}}}e_{i_{k-j}^*}(t)|^2, \ k\geq N-1, \ t\in [t_k, t_{k+1}),\\
V_G=\left\{ {\begin{array}{ll}
\sum_{i=1}^N(\tau_M-\eta_m)\int_{s_k}^{t}e^{2\alpha(s-t)}|\sqrt{G_i}C_i \dot x(s)|^2ds, \ k\geq N, \ t\in [t_k, t_{k+1}),\\
\sum_{i=1}^N(\tau_M-\eta_m)\int_{s_0}^{t}e^{2\alpha(s-t)}|\sqrt{G_i}C_i\dot x(s)|^2ds,\ t\in [{t_{N-1}}, t_{N}),
\end{array}} \right.
\end{array}
\end{equation}
where for $i=1,\dots, N$ \begin{equation}
\begin{array}{llll}\label{G:RR}
G_i=(N-1)Q_ie^{2\alpha [\tau_M+(N-2)(\tau_M-\eta_m)]}>0.
\end{array}\end{equation}
Here $V_e$ does not depend on $e_{i_{k}^*}(t_k)$.
Note that given $i=1,\dots,N$, $e_i$-term appears  $N-1$ times in $V_Q$
 for every $N$ intervals $[t_{k+j}, t_{k+j+1}), \ j=0,\dots, N-1$
 {(except of the interval with $i_{k+j}^*=i$).} This motivates $N-1$ in (\ref{G:RR})
 because $V_G$ is supposed to compensate $V_Q$ term.

As in the previous sections, the term $V_G$ is inserted to
cope with the delays in the reset conditions. It is continuous on
$[t_k,t_{k+1})$ and
 does not grow in the jumps (when $t=t_{k+1}$), since for  $k>
 N-1$ (cf. \eqref{Q:N})
\begin{equation}\label{Gi:RR:N}
\begin{array}{lll}
{V_G}_{|t=t_{k+1}}\!\!-\!{V_G}_{|t=t_{k+1}^-}
\leq -\sum_{i=1}^N(\tau_M-\eta_m)\int_{s_{k}}^{s_{k+1}}e^{2 \alpha
(s-t_{k+1})} |\sqrt {G_i}C_i\dot x(s)|^2 ds
\end{array}\end{equation}
and for $k=N-1$
\begin{equation}\label{Gi0:RR:N}
\begin{array}{lll}
{V_G}_{|t=t_{N}}-{V_G}_{|t=t_{N}^-}
\leq -\sum_{i=1}^N(\tau_M-\eta_m)\int_{s_0}^{s_{N}}e^{2 \alpha (s-t_{N})}
|{\sqrt{G_i}}C_i\dot x(s)|^2ds.
\end{array}\end{equation}
%
The term $V_Q$ grows in the jumps as follows:
$$
\begin{array}{lll}
{V_Q}_{|t=t_{k+1}}-{V_Q}_{|t=t_{k+1}^-}
&=&\sum_{ j=1}^{N-1}{t_{k+2}-t_{k+1}\over
j(\tau_M-\eta_m)}|{\sqrt{Q_{i_{k+1-j}^*}}} e_{i_{k+1-j}^*}(t_{k+1})|^2\\
&\leq&  \sum_{ j=0}^{N-2}{1\over j+1}|{\sqrt{Q_{i_{k-j}^*}}C_{i_{k-j}^*}}[x(s_{k-j})-x(s_{k+1})]|^2 \\
&\leq& \sum_{ j=0}^{N-2}(\tau_M-\eta_m)\int_{s_{k-j}}^{s_{k+1}} |\sqrt
{Q_{i_{k-j}^*}}C_{i_{k-j}^*}\dot x (s)|^2ds,\end{array}$$ where we
have used Jensen's inequality and the bound
\begin{equation}\begin{array}{lll}\label{skj1}
s_{k+1}-s_{k-j}&=&s_{k+1}-s_k+s_k-\dots +s_{k-j+1}-s_{k-j}\\
&\leq& (j+1)(\tau_M-\eta_m).\end{array}\end{equation}
 Since
$1\leq e^{2\alpha [\tau_M+(N-2)(\tau_M-\eta_m)]}e^{2 \alpha
(s_{k-j}-t_{k+1})}$ for  $j=0, \dots, N-2,$ we obtain
\begin{equation}\label{Q:RR:N}
\begin{array}{lll}
{V_Q}_{|t=t_{k+1}}-{V_Q}_{|t=t_{k+1}^-}
\leq \sum_{ j=0}^{N-2}(\tau_M-\eta_m) e^{2\alpha [\tau_M+(N-2)(\tau_M-\eta_m)]}\\
\hspace{3.8cm}\times\int_{s_{k-j}}^{s_{k+1}}e^{2 \alpha (s-t_{k+1})}
|{\sqrt{Q_{i_{k-j}^*}}} C_{i_{k-j}^*}\dot x (s)|^2 ds.
\end{array}\end{equation}
%

The following lemma gives sufficient conditions for
the $x$-ISS of (\ref{x:N}), (\ref{RR:ik:N}), (\ref{RRreset:N}) (see Appendix B for proof):

\begin{lem} \label{lemma:RR:N}
If there exist  positive constants $\alpha,$ $b$ and $V_e(t)$ of  (\ref{V:RR}) such that along (\ref{x:N}), (\ref{RR:ik:N}), (\ref{RRreset:N}) the inequality (\ref{inequality1:lemma:TOD}) is satisfied for
$k \geq N-1$.
Then the following bound holds along the solutions of (\ref{x:N}), (\ref{RR:ik:N}), (\ref{RRreset:N}):
\begin{equation}
\begin{array}{llll}\label{jump:tkRR:N}
 V_e(t_{k+1}) \leq e^{-2\alpha(t_{k+1}-t_{N-1})}V_e(t_{N-1})+\Psi_{k+1}\\
\hspace{1.7cm}+b \Delta^2\int_{t_{N-1}}^{t_{k+1}}e^{-2\alpha (t_{k+1}-s)}ds,\ k\geq N-1,
\end{array}\end{equation}
where
\begin{equation}\label{Psik:RR:N}
\begin{array}{llll}
\Psi_{k+1}=-(\tau_M-\eta_m)e^{2\alpha [\tau_M+(N-2)(\tau_M-\eta_m)]}\\
\hspace{1.4cm} \times \Big[\sum_{j=0}^{N-3}(N-2-j)
\int_{s_{k-j-1}}^{s_{k+1}}e^{2 \alpha (s-t_{k+1})}|\sqrt{Q_{i_{k-j}^*}}C_{i_{k-j}^*}\dot x(s)|^2ds\\
\hspace{1.8cm}+(N-1)\!\int_{s_k}^{s_{k+1}}\!e^{2 \alpha (s-t_{k+1})}
 |\sqrt{Q_{i_{k+1}^*}}C_{i_{k+1}^*}\dot x(s)|^2ds\Big]\!\leq 0.\end{array}\end{equation}
Moreover, for all $t\geq t_{N-1}$
\begin{equation}\label{lem:V:bound:RR:N}
\begin{array}{llll}
V(t, x_t, \dot x_t)\leq e^{-2\alpha (t-t_{N-1})} V_e(t_{N-1})
+{b \over  2\alpha}  \Delta^2, 
 \\
V_e(t_{N-1})=V(t_{N-1}, x_{t_{N-1}}, \dot
x_{t_{N-1}})+\sum_{i=1}^N\!|\sqrt{Q_i}e_i(t_{N-1})|^2. 
\end{array}\end{equation}
The latter inequality guarantees the $x$-ISS of  (\ref{x:N}), (\ref{RR:ik:N}), (\ref{RRreset:N}) for $t\geq t_{N-1}$.
\end{lem}

By using Lemma \ref{lemma:RR:N}, arguments
 of Theorem \ref{thm:ISS:TOD:N} and the fact that for $j=1,\dots N-1$
 $$\begin{array}{ll}{d\over dt} {t_{k+1}-t\over
j(\tau_M-\eta_m)}=-{1\over j(\tau_M-\eta_m)}
\leq -{1\over (N-1)(\tau_M-\eta_m)},\end{array}$$
  we arrive at the the following result:
\begin{thm}

\label{thm:ISS:RR:N}   Given $0 \leq \eta_m < \tau_M$ and
$\alpha >0$, assume that there exist positive scalar $b$, $n\times n$ matrices
$P>0$, $S_0>0$, $R_0>0$, $S_1>0$, $R_1>0$, $S_{12}$ and $n_i\times n_i$ matrices $Q_i>0$
$(i=1,\dots, N)$ such that
(\ref{LMI1:thm:ISS:TOD:N}) and
(\ref{LMI2:thm:ISS:TOD:N}) are feasible with
 $U_i={Q_i\over N-1}$, where $G_i$ is given by  \eqref{G:RR}.
Then for $N>2$ solutions of the hybrid system (\ref{x:N}), (\ref{RR:ik:N}), (\ref{RRreset:N}) satisfy
the bound (\ref{lem:V:bound:RR:N})
with $V(t, x_t, \dot x_t)$ given by (\ref{V:RR}),  meaning $x$-ISS for $t \geq t_{N-1}$.
For $N=2$ solutions of the hybrid system (\ref{x:N})-(\ref{if:N}) ((\ref{x:N}), (\ref{reset:N}), (\ref{RR:ik:N})) satisfy
the bound (\ref{lem:V:bound:N}) meaning $x$-ISS (for $t\geq t_0$).
Moreover, if the above LMIs are feasible with
$\alpha=0$, then the solution bounds  hold with a small enough $\alpha_0>0$.
\end{thm}

\begin{rk}\label{rem_ebound}
For $N=2$ and $\alpha=0$, the LMIs of Theorem
\ref{thm:ISS:TOD:N}   are more restrictive
than those of Theorem \ref{thm:ISS:RR:N}:
 \eqref{Q:Q12:N} of Theorem
\ref{thm:ISS:TOD:N} yields
$(N-1)U_i<Q_i<G_i$, whereas in Theorem \ref{thm:ISS:RR:N} we have
 $(N-1)^2U_i=(N-1)Q_i={G_i}$ that leads to
 larger $U_i$ for the same $G_i$. The latter helps for the feasibility of \eqref{LMI2:thm:ISS:TOD:N}, where $U_i>0$ appears on the main diagonal only (with minus).
However, 
Theorem \ref{thm:ISS:TOD:N} achieves ISS with respect to the full state $\col\{x, e\}$ and provides the solution bound
for $t\geq t_{0}$, while Theorem \ref{thm:ISS:RR:N}  guarantees only $x$-ISS.

Note that  Theorem \ref{thm:ISS:RR:N} under RR protocol guarantees
boundedness  of $e$ as well.
 Indeed, since $e(t_N)$ in \eqref{RRreset:N} depends on $x(0),\dots, x(t_N-\eta_N)$ and
 $t_N\leq N\tau_M$ (this can be verified similar to (\ref{skj1})),
relations \eqref{RRreset:N}  yield
$$|e_i(t)|^2\leq c'sup_{\theta\in[-N\tau_M,0]}|x(t+\theta)|^2, \ t\geq t_N$$
with some $c'>0$, which together with
(\ref{lem:V:bound:RR:N}) imply
$$|e_i(t)|^2\leq c''[e^{-2\alpha (t-t_{N-1})} V_e(t_{N-1})+\Delta^2] $$
for some $c''>0$ and all $t\geq t_N+N\tau_M$.
\end{rk}


\begin{rk}\label{polyt}
The LMIs of Theorems \ref{thm:ISS:TOD:N} and \ref{thm:ISS:RR:N}  are
affine in the system matrices. Therefore, in the case of system
matrices from an uncertain time-varying polytope
$$\begin{array}{lll}
\label{Omega}
  \Omega =\sum_{j=1}^{M}g_j(t)\Omega_j, \quad
  0\leq
 g_j(t)\leq 1,\\
 \sum_{j=1}^{M}g_j(t)=1, \quad
 \Omega_j =\left[\begin{array}{cccccc}A^{(j)}
 & B^{(j)} & D^{(j)} \end{array}\right],
 \end{array}
$$
one have to solve these LMIs simultaneously for all the $M$ vertices
$\Omega_j$, applying the same decision matrices.
\end{rk}

\section{Examples}\label{sec:examples}
\subsection{Example 1: uncertain inverted pendulum}

Consider an inverted pendulum mounted on a small car.
We focus on the stability analysis  in the absence of
disturbance. Following \cite{Geromel_IET07}, we assume that the friction coefficient
between the air and the car, $f_c$, and the air and the bar, $f_b$,
are not exactly known and time-varying: $f_c(t)\in [0.15, 0.25]$ and
$f_b(t)\in [0.15, 0.25]$. The linearized model can be written as
(\ref{sys}), where the matrices $A=E^{-1} A_f$ and $B=E^{-1} B_0$  are determined from
$$
\begin{array}{lllll}
E=\left[\begin{array}{cccc}
1 & 0 & 0 & 0 \\
0 & 1 & 0 & 0 \\
0 & 0 & 3/2 & -1/4 \\
0 & 0 &  -1/4 & 1/6
\end{array} \right],\\
A_f=\left[\begin{array}{cccc}
0 & 0 & 1 & 0 \\
0 & 0 & 0 & 1 \\
0 & 0 & -(f_c+f_b) & f_b/2 \\
0 & 5/2 &  f_b/2 & -f_b/3
\end{array} \right] \quad {\rm{and}} \quad  B_0=\begin{bmatrix}0\\0\\1\\0\end{bmatrix}.
\end{array}$$
%
Here $A$ belongs to  uncertain polytope,
 defined by   four
vertices corresponding to $f_c/f_b=0.15$ and $ f_c/f_b=0.25$.
%
The pendulum can be stabilized by a state feedback $u(t)=Kx(t)$, where $x=[x_1,x_2,x_3,x_4]^T$, with
the gain
\begin{equation} \label{gainK} K=[11.2062 \ \ -128.8597 \ \
10.7823 \ \ -22.2629].\end{equation}
In this model, $x_1$ and  $x_2$
represent cart position and velocity, whereas $x_3,$ $x_4$ represent pendulum angle from vertical and its angular velocity
respectively.
In practice  $x_1,$ $x_2$ and $x_3$, $x_4$ (presenting spatially distributed
components of the state of the pendulum-cart system)
are not accessible simultaneously.
Suppose that the state variables  
are not accessible simultaneously. 
Consider first $N=2$ and
$$
C_1 = \left[
\begin{array}{ccccccccc}
       1 &0 &0 & 0 \\
       0&1& 0 & 0
     \end{array}\right],\
C_2 = \left[
\begin{array}{ccccccccc}
       0 & 0 &1 & 0 \\
       0&0 & 0 & 1
      \end{array}\right].
$$
The applied controller gain $K$ has the following blocks:
$$
K_1 =\left[
\begin{array}{ccccccccc}
       11.2062  &-128.8597
      \end{array}\right], \ \
K_2 = \left[
\begin{array}{ccccccccc}
       10.7823 & -22.2629
      \end{array}\right].
$$

%
%
For the  values of $\eta_m$ given in Table 1, we apply Theorems
\ref{thm:ISS:TOD:N} and \ref{thm:ISS:RR:N} with  $\alpha=0$, $b=0$  via Remark
\ref{polyt} and find the maximum values of $\tau_M= \MATI + \MAD$
that preserve the stability of the hybrid system
(\ref{x:N})-(\ref{reset:N})
with $\omega(t)=0$ with respect to $x$. 
%
%
From Table 1, it is observed that under TOD or RR protocol the conditions of Theorem \ref{thm:ISS:RR:N}
 possess less decision variables, and stabilize
the system for larger $\tau_M$ than the results in \cite{RR_SCL} under RR protocol.
 Moreover, when $\eta_m > {\tau_M \over 2} (\eta_m=0.02, 0.04)$,
our method is still feasible (communication delays are larger than the sampling intervals).
The computational time for solving the LMIs (in seconds) under the TOD protocol is essentially less
than that under RR protocol in \cite{RR_SCL} (till  36\% decrease).

\begin{table}[h]
\label{stability:RR:TOD:cart:poly}
\begin{center}
\caption{Example 1 (N=2): max. value of $\tau_M= \MATI + \MAD$}
\begin{tabular}{|c|c|c|c|c|c|c|c|c|} \hline
$\tau_M $ ${\large \backslash}$ $\eta_m$      &0       &0.005    & 0.01    & 0.02   & 0.04 & Decision\\
&     &   &    &    &  & variables\\
\hline \hline \cite{RR_SCL} (RR) & 0.023& 0.026& 0.029&  0.035& 0.046 & 146\\
\hline Theorem \ref{thm:ISS:TOD:N} (TOD)  & 0.014 & 0.018 & 0.021&0.029 &  0.044 & 84\\
\hline Theorem \ref{thm:ISS:RR:N} (TOD/RR)   & 0.025& 0.028& 0.031&  0.036& 0.047&72\\
 \hline
\end{tabular}
\end{center}
\end{table}


Consider next $N=4$, where  $C_1,\dots, C_4$ are the rows of $I_4$
and $K_1,\dots, K_4$ are the entries of $K$ given by \eqref{gainK}.
Here  the maximum values of $\tau_M$ 
that preserve
the stability of 
(\ref{x:N})-(\ref{reset:N}) with $\omega(t)=0$ with respect to $x$
are given in Table 2. Also here Theorem \ref{thm:ISS:RR:N} leads to
less conservative results than Theorem \ref{thm:ISS:TOD:N}.
\begin{table}[h]
\label{stability:RR:TOD:cart:poly}
\begin{center}
\caption{Example 1 (N=4): max. value of $\tau_M= \MATI + \MAD$}
 \begin{tabular}{|c|c|c|c|c|c|c|c|c|}\hline
$\tau_M $ ${\large \backslash}$ $\eta_m$      &0       &0.01\\
\hline Theorem \ref{thm:ISS:TOD:N} (TOD)  & 0.003 & 0.012 \\
\hline Theorem \ref{thm:ISS:RR:N} (RR)   & 0.006& 0.015\\
\hline
\end{tabular}
\end{center}
\end{table}
%

\subsection{Example 2: batch reactor}
 We illustrate the efficiency of the given conditions on the
 example of a batch reactor under the dynamic output feedback (see e.g., \cite{HeemelsTAC2010}), where $N=2$ and
$$\begin{array}{lllll}
A=\left[
\begin{array}{ccccccc}
1.380  & -0.208  & 6.715   & -5.676 \\
-0.581 & -4.2902 & 0       & 0.675 \\
1.067  & 4.273   & -6.654  & 5.893\\
0.048  & 4.273   & 1.343   & -2.104\\
\end{array} \right],
\end{array}$$
$$\begin{array}{lllll}
B=\left[ \begin{array}{ccccccc}
0      & 0  \\
5.679  & 0 \\
1.136  & -3.146  \\
1.136  & 0\\
\end{array} \right], \ C=\small\left[
\begin{array}{c} C_1
\\
\hline C_2
                          \end{array} \right]=\left[ \begin{array}{cccc}
1      & 0  & 1& -1
\\
\hline
 0      & 1  & 0& 0
\end{array} \right],\\
\left[ \begin{array}{ccccccc}
  A_c     &    \vline     & B_c \\
\hline
  C_c    &    \vline     & D_c \\
\end{array} \right]=
\left[ \begin{array}{ccccccc}
  \left. \begin{array}{cc}
0      & 0  \\
0       & 0
\end{array} \right.    &    \vline     &   \left. \begin{array}{cc}
0      & 1  \\
1       & 0
\end{array} \right.  \\
\hline
 \left. \begin{array}{cc}
-2      & 0  \\
0       & 8
\end{array} \right.    &    \vline     & \left. \begin{array}{cc}
0      & -2  \\
5       & 0
\end{array} \right. \\
\end{array} \right].
\end{array}$$
For the values of $\eta_m$ given in Table 3, we apply Theorems
\ref{thm:ISS:TOD:N} and \ref{thm:ISS:RR:N} with  $\alpha=0$, $b=0$
and find the maximum values of $\tau_M= \MATI + \MAD$ that
preserve the stability of the hybrid system
(\ref{x:N})-(\ref{reset:N}) with $\omega(t)=0$ with respect to $x$. From
Table 3 it is seen that the results of our method essentially
improve the results in \cite{HeemelsTAC2010}, and are more conservative than those obtained via the discrete-time approach. 
Recently in \cite{HeemelsSOSAut12} the same result $\tau_M=0.035$ as
ours in Theorem \ref{thm:ISS:RR:N} for $\eta_m=0,$ $\MAD=0.01$ has been achieved. In \cite{HeemelsSOSAut12} the sum of squares method is developed in the framework of
hybrid system approach. We note that the sum of squares method has
not been applied yet to ISS. Moreover, our conditions are simple
LMIs with a
fewer decision variables. 
When $\eta_m > {\tau_M \over 2} (\eta_m=0.03, 0.04)$, our method is
still feasible
 (communication delays are larger than the sampling intervals).
The computational time under the TOD protocol is essentially less
than that under RR protocol in \cite{RR_SCL} (till  32\%
decrease).
\begin{table}[h]
\label{lower:delay}
\begin{center}
\caption{Example 2: max. value of $\tau_M= \MATI + \MAD$ for
different $\eta_m$ } \begin{tabular}{|c|c|c|c|c|c|c|c|} \hline
$\tau_M $ ${\large \backslash}$ $\eta_m$      &0       &0.004    & 0.01   & 0.02   & 0.03    & 0.04\\
\hline \hline  \cite{HeemelsTAC2010}($\MAD=0.004$)  & 0.0108& 0.0133& -&  -& -& -\\
\hline  \cite{donkers_2011}($\MAD=0.03$) & 0.069& 0.069&0.069& 0.069& 0.069& -\\
\hline Theorem \ref{thm:ISS:TOD:N} (TOD)  & 0.019& 0.022&  0.027&  0.034& 0.042& 0.050\\
\hline Theorem \ref{thm:ISS:RR:N} (TOD/RR)  & 0.035& 0.037&  0.041&  0.047& 0.053& 0.059\\
\hline \cite{RR_SCL} (RR)  & 0.042& 0.044&  0.048&  0.053& 0.058& 0.063\\
 \hline
\end{tabular}
\end{center}
\end{table}



\section{Conclusions}
In this paper, a time-delay approach has been developed for the
ISS of NCS with scheduling protocols, 
variable transmission delays and variable sampling
intervals. 
A novel hybrid system model with time-varying delays in the
continuous dynamics and in the reset equations is introduced and a new
Lyapunov-Krasovskii method is developed.
The ISS conditions 
of the delayed hybrid system
are derived in terms of LMIs.
Differently from the existing (hybrid and discrete-time)
 methods on the stabilization of
NCS with scheduling protocols,
the time-delay approach allows non-small network-induced delay
(which is
not smaller than the sampling interval). 
Future work will involve  consideration of  more general NCS models,
including packet dropouts, packet
disordering, quantization and scheduling protocols for the actuator nodes.

\appendix
\section{Proof of Theorem \ref{thm:ISS:TOD:N}}
\begin{proof}
Consider $t\in[t_k, t_{k+1}), \ k\in {\Z_+}$ and define
$\xi_i(t)=\col\{x(t), x(t-\eta_m), x(t-\tau(t)), x(t-\tau_M), e_1(t),
\cdots, e_j(t), \cdots, e_N(t), \omega(t)\}$ with $i=i_k^* \in \N,$ $j \neq i$. Differentiating
$V_e(t)$ along (\ref{x:N}) and applying Jensen's
inequality, we have
$$ \begin{array}{lllllll}
\eta_m \int_{t-\eta_m}^{t} \dot
x^T(s)R_0 \dot x(s) ds \geq  \int_{t-\eta_m }^{t} \dot x^T(s)ds R_0 \int_{t-\eta_m}^{t}\dot x(s)ds\\
\hspace{3.8cm}=\xi_i^T(t)(F_2^i)^TR_0 F_2^i\xi_i(t),\\
-(\tau_M-\eta_m)\int_{t-\tau_M}^{t-\eta_m}\dot x^T(s)R_1 \dot x(s) ds\\
=-(\tau_M-\eta_m)\int_{t-\tau(t)}^{t-\eta_m}\dot x^T(s)R_1 \dot x(s) ds-(\tau_M-\eta_m)\int_{t-\tau_M}^{t-\tau(t)}\dot x^T(s)R_1 \dot x(s) ds\\
\leq -{\tau_M-\eta_m \over \tau(t)-\eta_m}\xi_i^T(t)\Big[[I_{n}\ 0_{n \times n]}F^i\Big]^{T}R_1 [I_{n }\ 0_{n \times n}]F^i\xi_i(t)\\
\hspace{.35cm}-{\tau_M-\eta_m \over \tau_M-\tau(t)}\xi_i^T(t)\Big[[0_{n \times n}\ I_{n}]F^i\Big]^{T}R_1 [0_{n \times n}\ I_{n}]F^i\xi_i(t)\\
\leq -\xi_i^T(t)(F^i)^T\Phi F^i\xi_i(t).
\end{array}$$
The latter inequality holds
if (\ref{LMI1:thm:ISS:TOD:N}) is feasible \cite{Park2011}. Then
$$\begin{array}{llll}
\dot {V_e}(t)+2\alpha V_e(t)-{1\over \tau_M-\eta_m} \sum_{l=1, l \neq
i}^N|\sqrt{U_l}e_l(t)|^2-2\alpha|\sqrt{Q_{i}}e_{i}(t)|^2- b |\omega(t)|^2\\
\leq \xi_i^T(t)[\Sigma_i+ \Xi_i^T H  \Xi_i- (F^i)^T\Phi
  F^i e^{-2\alpha \tau_M}]\xi_i(t)\leq 0,
\end{array}$$
if $\Sigma_i+ \Xi_i^T H  \Xi_i- (F^i)^T\Phi   F^i e^{-2\alpha
\tau_M}<0$, i.e., by Schur complement, if
(\ref{LMI2:thm:ISS:TOD:N}) is feasible.
 Thus due to Lemma \ref{lemma:TOD:N},
inequalities (\ref{Q:Q12:N}), (\ref{LMI1:thm:ISS:TOD:N}) and
(\ref{LMI2:thm:ISS:TOD:N}) imply (\ref{lem:V:bound:N}). \end{proof}

\section{Proof of Lemma \ref{lemma:RR:N}}
\begin{proof}
Since $|\omega(t)| \leq \Delta$,
(\ref{inequality1:lemma:TOD}) implies
\begin{equation}\begin{array}{llll}\label{lem1:RR:N}
 V(t, x_t, {\dot x}_t)&\leq& V_e(t)\\
 &\leq& e^{-2\alpha(t-t_k)}V_e(t_k)+ b\Delta^2\int_{t_k}^te^{-2\alpha (t-s)}ds,
\ t\in[t_k, t_{k+1}).
\end{array}\end{equation}
Note that
$
V_e(t_{k+1})= \tilde V_{|t=t_{k+1}}+V_{Q_{|t=t_{k+1}}}+V_{G_{|t=t_{k+1}}}.
$
Taking into
account  \eqref{RRreset:N} and the relations
$\tilde V_{t=t_{k+1}}=\tilde V_{t=t_{k+1}^- }$, $e(t_{k+1}^-)=e(t_k),$
 we obtain due to (\ref{G:RR}), (\ref{Gi:RR:N}) and (\ref{Q:RR:N}) for $k>N-1$
\begin{equation}\label{Tetk:RR:N}
\begin{array}{llllll}
\Theta_{k+1}\bydef V_e(t_{k+1})-V_e(t_{k+1}^-)\\
\hspace{.9cm}=[V_Q+V_G]_{t=t_{k+1}}-[V_Q+V_G]_{t=t_{k+1}^- }\\
\hspace{.9cm}\leq \!\!(\tau_M\!\!-\!\!\eta_m)e^{2\alpha [\!\tau_M\!+\!(N\!-\!2)(\tau_M\!-\!\eta_m)]}\!\Big[\!\sum_{ j=0}^{N-2}\int_{s_{k-j}}^{s_{k+1}}e^{2 \alpha(s\!-\!t_{k+1})}|\!\sqrt {Q_{i_{k-j}}^*}C_{i_{k-j}}^*\dot x (s)|^2ds\\
\hspace{1.2cm}-\sum_{i=1}^N(N-1)\int_{s_{k}}^{s_{k+1}}e^{2 \alpha
(s-t_{k+1})} |\sqrt {Q_i}C_i\dot x(s)|^2 ds\Big],
\end{array}\end{equation}
whereas for $k=N-1$ due to (\ref{Gi0:RR:N}) and (\ref{Q:RR:N})
$$
\begin{array}{llllll}
\Theta_{N} \leq \sum_{ j=0}^{N-2}(\tau_M-\eta_m)e^{2\alpha [\tau_M+(N-2)(\tau_M-\eta_m)]}\\
\hspace{1.2cm} \times \int_{s_{N-1-j}}^{s_{N}}e^{2 \alpha (s-t_{N})}|\sqrt {Q_{i_{N-1}-j}^*} C_{i_{N-1}-j}^*\dot x(s)|^2ds\\
\hspace{1cm} -\sum_{i=1}^N(\tau_M-\eta_m)\int_{s_0}^{s_{N}}e^{2 \alpha (s-t_{N})}
|{\sqrt{G_i}}C_i\dot x(s)|^2ds\\
\hspace{.7cm}\leq -(\tau_M-\eta_m)e^{2\alpha [\tau_M+(N-2)(\tau_M-\eta_m)]}\int_{s_0}^{s_{N}}e^{2 \alpha (s-t_{N})}\\
\hspace{1.1cm} \times[(N-2)\sum_{i=1}^N|{\sqrt{Q_i}}C_i\dot x(s)|^2+|{\sqrt{Q_l}}C_l\dot x(s)|^2_{|l=i^*_{N}}]ds.
\end{array}$$
We will prove  (\ref{jump:tkRR:N}) by induction. For $k=N-1$ we have
$$\begin{array}{llllll}
 V_e(t_{N})\leq \Theta_{N}+V_e(t_{N}^-)\\
\hspace{1.1cm} \leq -(\tau_M-\eta_m)e^{2\alpha [\tau_M+(N-2)(\tau_M-\eta_m)]}\\
\hspace{1.5cm}\times \int_{s_0}^{s_{N}}e^{2 \alpha (s-t_{N})}[(N\!-\!2)
 \sum_{i=1}^N|{\sqrt{Q_i}}C_i\dot x(s)|^2+|{\sqrt{Q_l}}C_l\dot
 x(s)|^2_{|l=i^*_{N}}]ds\\
\hspace{1.5cm}+e^{-2\alpha(t_{N}-t_{N-1})}V_e(t_{N-1})+b\Delta^2\int_{t_{N-1}}^{t_{N}}e^{-2\alpha (t_{N}-s)}ds,
\end{array}$$
which implies (\ref{jump:tkRR:N}).

Assume that (\ref{jump:tkRR:N}) holds for $k-1$ ($k\geq N-1$):
$$\begin{array}{lll} V_e(t_{k}) \leq e^{-2\alpha(t_{k}-t_{N-1})}V_e(t_{N-1})+\Psi_{k}+b\Delta^2\int_{t_{N-1}}^{t_{k}}e^{-2\alpha (t_{k}-s)}ds.\end{array}$$
Then due to (\ref{lem1:RR:N}) for $t=t_{k+1}^-$ we obtain
$$\begin{array}{llllll}
 V_e(t_{k+1})\leq \Theta_{k+1}+e^{-2\alpha(t_{k+1}-t_{k})}\Psi_k+e^{-2\alpha(t_{k+1}-t_{N-1})}V_e(t_{N-1})\\
\hspace{1.7cm} +b\Delta^2\int_{t_{N-1}}^{t_{k+1}}e^{-2\alpha (t_{k+1}-s)}ds. \end{array}$$
We have
$$\begin{array}{llllll}
e^{-2\alpha(t_{k+1}-t_{k})}\Psi_k
=-(\tau_M-\eta_m)e^{2\alpha [\tau_M+(N-2)(\tau_M-\eta_m)]}\\
\quad \times \Big[\sum_{l=0}^{N-3}(N-2-l)\int_{s_{k-l-2}}^{s_{k}}e^{2 \alpha (s-t_{k+1})}\times |\sqrt{Q_{i_{k-1-l}^*}}C_{i_{k-1-l}^*}\dot x(s)|^2ds\\
\quad \quad+(N-1)\int_{s_{k-1}}^{s_{k}}e^{2 \alpha (s-t_{k+1})}
|\sqrt{Q_{i_{k}^*}}C_{i_{k}^*}\dot x(s)|^2ds\Big] \\
= -(\tau_M-\eta_m)e^{2\alpha [\tau_M+(N-2)(\tau_M-\eta_m)]}\Big[\sum_{j=0}^{N-2}(N-1-j)\\
 \quad \quad \times \int_{s_{k-j-1}}^{s_{k}}e^{2 \alpha (s-t_{k+1})}|\sqrt{Q_{i_{k-j}^*}}C_{i_{k-j}^*}\dot x(s)|^2ds\Big].
 \end{array}$$
Then, taking into account
(\ref{Tetk:RR:N}), we find
$$\begin{array}{llllll}
\Theta_{k+1}+e^{-2\alpha(t_{k+1}-t_{k})}\Psi_k
\leq(\tau_M-\eta_m)e^{2\alpha [\tau_M+(N-2)(\tau_M-\eta_m)]}\\
\quad \times \Big[\sum_{ j=0}^{N-2}\int_{s_{k-j-1}}^{s_{k+1}}e^{2 \alpha
(s-t_{k+1})}|\sqrt {Q_{i_{k-j}}^*}C_{i_{k-j}}^*\dot x (s)|^2ds\\
\quad \quad-\sum_{i=1}^N(N-1)\int_{s_{k}}^{s_{k+1}}e^{2 \alpha(s-t_{k+1})} |\sqrt {Q_i}C_i\dot x(s)|^2 ds\\
\quad \quad -\sum_{j=0}^{N-2}(N-1-j)
 \int_{s_{k-j-1}}^{s_{k}}e^{2 \alpha (s-t_{k+1})}\times |\sqrt{Q_{i_{k-j}^*}}C_{i_{k-j}^*}\dot x(s)|^2ds\Big]\\
 \leq-(\tau_M-\eta_m)e^{2\alpha [\tau_M+(N-2)(\tau_M-\eta_m)]}\\
\quad \times \Big[\sum_{j=0}^{N-2}(N-2-j)\int_{s_{k-j-1}}^{s_{k+1}}e^{2 \alpha (s-t_{k+1})}\times |\sqrt{Q_{i_{k-j}^*}}C_{i_{k-j}^*}\dot x(s)|^2ds\\
\quad+(N\!-\!1)\int_{s_k}^{s_{k+1}}\!\!e^{2 \alpha (s-t_{k+1})}
   |\sqrt{Q_{i_{k+1}^*}}C_{i_{k+1}^*}\dot x(s)|^2ds\Big]\\
= \Psi_{k+1},
\end{array}$$
which implies (\ref{jump:tkRR:N}).
Hence, (\ref{jump:tkRR:N})  and (\ref{lem1:RR:N}) yield
(\ref{lem:V:bound:RR:N}). \end{proof}

\end{document}